\documentclass{article}

\usepackage{fullpage}
\usepackage{amsmath}
\usepackage[pdftex]{graphicx,color}
\usepackage{epsfig}
\usepackage{url}
\usepackage{microtype}
\usepackage{cite}
\usepackage{mathtools}
\usepackage{amsthm}
\usepackage{tikz}
\usetikzlibrary{shapes.geometric,fit}

\newcommand{\op}[1]{\textsc{#1}}

\newcommand{\tcr}{}

\newcommand{\rk}{\tcr{21}} 
\newcommand{\rkk}{\tcr{42}} 
\newcommand{\rkkk}{\tcr{63}} 
\newcommand{\hp}{\tcr{36}} 
\newcommand{\hpac}{\tcr{8}} 
\newcommand{\gpc}{\tcr{6}} 
\newcommand{\nib}{\tcr{6}} 
\newcommand{\gpcnib}{\tcr{36}} 
\newcommand{\nbn}{\tcr{7}} 
\newcommand{\nbpe}{\tcr{12}} 
\newcommand{\nbp}{\tcr{84}} 
\newcommand{\emmc}{\tcr{117}} 
\newcommand{\emacb}{\tcr{102}} 
\newcommand{\emac}{\tcr{101}} 
\newcommand{\mhc}{\tcr{21}} 
\newcommand{\dcmc}{\tcr{27}} 
\newcommand{\dcac}{\tcr{26}} 
\newcommand{\dac}{\tcr{127}} 
\newcommand{\dmc}{\tcr{144}} 

\newtheorem{theorem}{Theorem}
\newtheorem{lemma}[theorem]{Lemma}
\newtheorem{corollary}[theorem]{Corollary}

\begin{document}

\title{Improved Upper Bounds for Pairing Heaps\thanks{Part of
this work appeared in preliminary form in \cite{swat00} and \cite{thesis}.}}

\author{John Iacono\thanks{Research supported by NSF grants
CCR-9732689, CCR-0430849, CCF-1018370, and by an Alfred P.~Sloan Fellowship. Research was partially completed
at Rutgers, The State University of New Jersey---New Brunswick.}}

\date{Department of Computer Science and Engineering\\
New York University Polytechnic School of Engineering\\
}

\maketitle

\begin{abstract}

Pairing heaps are shown to have constant amortized time
\op{Insert} and \op{Meld}, thus showing that pairing heaps
have the same amortized runtimes as Fibonacci heaps for all
operations but \op{Decrease-Key}.

\end{abstract}

\section{Introduction}

\begin{figure}
\begin{center}
\begin{tabular}{cc||cc}
& & \op{Insert} and \op{Meld}  & \op{Decrease-Key} \\ \hline \hline
1986 &Original pairing heap \cite{pair} & $O(\log n)$  & $O(\log n)$ \\
1987&Stasko and Vitter (variant) \cite{vitter} & $O(1)$ & Forbidden \\
1999&Fredman \cite{noconst} & $O(\log n)$  & $\Omega(\log \log n)$ \\
2001&This Paper & $O(1)$ & $O(\log n)$\\
2005&Pettie \cite{pettie} & $O(2^{2 \sqrt{\log \log n}})$ & 
$O(2^{2 \sqrt{\log \log n}})$ \\
2010&Elmasry  (variant)\cite{DBLP:conf/esa/Elmasry10} & $O(1)$& $O(\log \log n)$\\
\end{tabular}
\end{center}
\caption{Summary of previous results for pairing heaps and their variants, in chronological order of first appearance. All previous results support
$O(1)$ amortized \op{Make-Heap} and $O(\log n)$ amortized \op{Extract-Min}. All results are amortized,
so the runtimes of the individual operations can not be mixed among the
different results. The result of Fredman's \cite{noconst} gives an amortized lower
bound for \op{Decrease-Key}, given amortized upper bounds for other
operations. The result of Stasko and Vitter is for a variant of
pairing heaps and does not allow the \op{Decrease-Key} operation.}
\label{f1}
\end{figure}

Pairing heaps were introduced in \cite{pair}
as a priority queue data structure modeled after
splay trees \cite{splay}. As in splay trees, 
pairing heaps do not augment the nodes of the heap
with any information, and use simple local restructuring
heuristics to perform all operations. Such structures are referred to as \emph{self-adjusting}. They are easy to
code and empirically perform well \cite{vitter}.

While the study of splay trees has been focused on their ability to
quickly execute various distributions of operations
\cite{splay,cole,cole2,algorithmica04,focs04,soda01,splayseq,mbed,amr,blum},
the study of pairing heaps remains stuck at a much earlier stage ---
tight amortized bounds on the runtimes of all operations in terms of
the heap size remain unknown. There is also small body of work studying how pairing heaps work on particular distributions of operations
\cite{amr2,algorithmica04a,swat00}.

On the practical side, pairing heaps are in use. They are covered in
some elementary and intermediate level data structures texts
\cite{weiss1,weiss2,weiss3,weiss4}. Additionally, one can find on
the internet code that has been developed to implement pairing heaps in
various languages. Pairing heaps were part of the pre-STL GNU C++
Library, and thus were distributed widely. 

The theoretically leading non-self adjusting priority queue is the
Fibonacci heap \cite{fib}, and pairing heaps are designed to support
the same set of operations as Fibonacci heaps:

\begin{itemize}

\item{$h$=\op{Make-Heap}()}: Returns an identifier $h$ of a new empty heap.

\item{$x$=\op{Extract-Min}($h$)}: Removes and returns $x$, the minimum element in 
heap $h$.

\item{$p$=\op{Insert}($h, x$)}: Inserts $x$ into heap $h$, and returns an identifier
$p$ that can be used to manipulate $x$ in the future.

\item{\op{Delete}($h,p$)}: Removes the item in $h$ identified by $p$.

\item{\op{Decrease-Key}($h, p, \Delta$)}: Decreases the key value of the
item in $h$ identified by $p$ by a nonnegative amount~$\Delta$.

\item{$h$=\op{Meld}($h_1,h_2$)}: Combines the contents of heaps $h_1$ and $h_2$
into a new heap with returned identifier $h$. The identifiers to~$h_1$ and~$h_2$
are no longer valid.

\item{$x$=\op{Find-Min}($h$)}: Returns the minimum element of the heap $h$.

\end{itemize}

In the original pairing heap paper \cite{pair}, all operations except
the constant-amortized-time \op{Make-Heap} and \op{Find-Min} were shown to take $O(\log n)$
amortized time (see Figure~\ref{f1} for a comparison of the known bounds on pairing heaps and their variants
). It was conjectured in \cite{pair} and empirical evidence was
presented by Stasko and Vitter \cite{vitter} that pairing heaps share
the same amortized cost per operation as Fibonacci heaps, which have
$O(1)$ \op{Decrease-Key}, \op{Insert} and \op{Meld} operations.  However, this
possibility was eliminated when it was shown by Fredman \cite{noconst}
that the amortized cost of \op{Decrease-Key} must be $\Omega(\log
\log n)$ amortized, given $O(\log n)$ amortized costs for the other operations. 
Recently, Pettie has produced a new analysis that focuses on the
\op{Decrease-Key} operation, where he proves a $O(2^{2 \sqrt{\log \log n}})$
amortized bound on \op{Insert}, \op{Meld}, and \op{Decrease-Key} while
retaining a $O(\log n)$ upper bound on $\op{Extract-Min}$. 
The asymptotic cost of \op{Decrease-Key} remains unknown as Fredman's $\Omega(\log \log n)$ amortized lower bound and
Pettie's $O(2^{2 \sqrt{\log \log n}})$ amortized upper bound are the best known amortized bounds.

In this work we present a new analysis of pairing heaps
that proves, with the exception of the
\op{Decrease-Key} operation, pairing heaps share the same asymptotic
amortized runtimes per operation as Fibonacci heaps. Specifically, we show the
amortized cost of \op{Extract-Min}, \op{Delete}, and \op{Decrease-Key}
is $O(\log n)$ and the amortized cost of \op{Make-Heap}, \op{Insert}, and \op{Meld}
is $O(1)$.  Thus, compared to
the original analysis in \cite{pair} the
amortized upper bound of $O(\log n)$ for the \op{Insert} and \op{Meld}
operations 
 is improved to $O(1)$. 
Compared to the analysis of Pettie, the $O(2^{2 \sqrt{\log \log n}})$ amortized
bounds for \op{Insert} and \op{Meld} are improved to a constant,
while Pettie's $O(2^{2 \sqrt{\log \log n}})$ amortized bound for \op{Decrease-Key}
is tighter than our $O(\log n)$ amortized bound.

It should
be noted that Stasko and Vitter in \cite{vitter} introduced a variant
of pairing heaps, the auxiliary twopass method, and proved that this
structure supportes constant amortized time \op{Insert}. However, their analysis
explicitly forbade the \op{Decrease-Key} operation. Elmasry \cite{DBLP:conf/esa/Elmasry10} invented a different variant of pairing heaps that obtains constant amortized \op{Insert} and amortized $O(\log \log n)$ \op{Decrease-Key}; however Fredman's $\Omega(\log \log n)$ amortized lower bound on \op{Decrease-Key} does not apply to Elmasry's variant as they do not confirm to his model of \emph{generalized pairing heaps}.

Here it is shown that pairing heaps are shown to have constant amortized time
\op{Insert} and \op{Meld}, thus showing that pairing heaps
have the same amortized runtimes as Fibonacci heaps for all
operations but \op{Decrease-Key}. The \op{Decrease-Key} operation is allowed at an amortized cost of $O(\log n)$.

\section{Pairing Heaps}

\newcommand{\ssb}[1]{\begin{center}\scalebox{0.9}{#1}\end{center}}

\newcommand{\bt}{
\begin{tikzpicture}[child anchor=north,thick,scale=0.70]
    \tikzstyle{a}=[circle]
    \tikzstyle{b}=[circle,draw,fill=blue!20]
        \tikzstyle{c}=[isosceles triangle,draw,fill=yellow,shape border rotate=90]]
}

\newcommand{\pair}[2]{
 \draw[red,<->,draw, thick,shorten >=2pt, shorten <=2pt] (#1) .. controls +(up:1cm) and +(up:1cm) .. node[above,sloped] {pair} (#2);
}

\begin{figure}
\begin{center}
\begin{tabular}{p{2in}p{2in}}
\ssb{
\bt
\tikzstyle{level 1}=[sibling distance=10mm]
    \node[b] {1}
        child { node[b] {4} child{ node[c]{}} }
        child { node[b] {3} child{ node[c]{}} }
        child { node[b] {7} child{ node[c]{}} }
        child { node[b] {2}  child{ node[c]{}}}
        child { node[b] {5} child{ node[c]{}} }
        child { node[b] {9} child{ node[c]{}} }
        child { node[b] {8} child{ node[c]{}} }
        child { node[b] {6} child{ node[c]{}} }
    ;
\end{tikzpicture}
}
&
\ssb{
\bt
\tikzstyle{level 1}=[sibling distance=10mm]
    \node(r)[a] {}
        child { node[b] {4} edge from parent[draw=none] child{ node[c]{}} }
        child { node[b] {3} edge from parent[draw=none] child{ node[c]{}}}
        child { node[b] {7} edge from parent[draw=none] child{ node[c]{}}}
        child { node[b] {2} edge from parent[draw=none] child{ node[c]{}}}
        child { node[b] {5} edge from parent[draw=none] child{ node[c]{}}}
        child { node[b] {9} edge from parent[draw=none] child{ node[c]{}}}
        child { node[b] {8} edge from parent[draw=none] child{ node[c]{}}}
        child { node[b] {6} edge from parent[draw=none] child{ node[c]{}}} 
	;
	\pair{r-1}{r-2}
	\pair{r-3}{r-4}
	\pair{r-5}{r-6}
	\pair{r-7}{r-8}
\end{tikzpicture}
}
\\
(a) Remove the root. & (b) The first pairing pass groups the nodes in pairs, and pairs them.
\\
\ssb{
\bt
\tikzstyle{level 1}=[sibling distance=20mm]
\tikzstyle{level 2}=[sibling distance=10mm]
    \node(r)[a] {}
        child { node[b] {2}  edge from parent[draw=none]  child { node[b] {4} child{ node[c]{}} } child{ node[c]{}}}
        child { node[b] {3} edge from parent[draw=none]  child { node[b] {7} child{ node[c]{}} } child{ node[c]{}}}
        child { node[b] {5} edge from parent[draw=none]  child { node[b] {9} child{ node[c]{}} } child{ node[c]{}}}
        child { node[b] {6} edge from parent[draw=none]  child { node[b] {8} child{ node[c]{}} } child{ node[c]{}}}
    ;
    	\pair{r-3}{r-4}
\end{tikzpicture}
}
&
\ssb{
\bt
\tikzstyle{level 1}=[sibling distance=25mm]
\tikzstyle{level 2}=[sibling distance=10mm]
    \node(r)[a] {}
        child {node[b] {2} edge from parent[draw=none]  child { node[b] {4} child{ node[c]{}} } child{ node[c]{}}}
        child { node[b] {3} edge from parent[draw=none]  child { node[b] {7} child{ node[c]{}} } child{ node[c]{}}}
        child { node[b] {5} edge from parent[draw=none]  
          child { [sibling distance=7mm] node[b] {6}  child { node[b] {8} child{ node[c]{}} } child{ node[c]{}}}
          child { node[b] {9} child{ node[c]{}} } child{ node[c]{}}}        
         ;
         	\pair{r-2}{r-3}
\end{tikzpicture}
}
\\
(c) The second pairing pass incrementally pairs the right two nodes until a single tree is formed. 
& (d) Second pairing pass, continued.
\\
\ssb{
\bt
\tikzstyle{level 1}=[sibling distance=40mm]
\tikzstyle{level 2}=[sibling distance=20mm]
    \node(r)[a] {}
        child {  node[b]{2} edge from parent[draw=none]  child { node[b] {4} child{ node[c]{}} } child{ node[c]{}}}
        child { node[b] {3} edge from parent[draw=none]  
                   child { [sibling distance=10mm] node[b] {5}   
          child {[sibling distance=7mm] node[b] {6}  child { node[b] {8} child{ node[c]{}} } child{ node[c]{}}}
          child { node[b] {9} child{ node[c]{}} } child{ node[c]{}}}        
           child { node[b] {7} child{ node[c]{}} } child{ node[c]{}}}
         ;
   \pair{r-1}{r-2}
\end{tikzpicture}
}
&
\ssb{
\bt
    \node[a] {}
        child {[sibling distance=20mm] node[b] {2} edge from parent[draw=none]  child { node[b] {4} child{ node[c]{}} } child{ node[c]{}}
        child { [sibling distance=20mm] node[b] {3}  
                   child { [sibling distance=10mm] node[b] {5}   
          child {[sibling distance=7mm] node[b] {6}  child { node[b] {8} child{ node[c]{}} } child{ node[c]{}}}
          child { node[b] {9} child{ node[c]{}} } child{ node[c]{}}}        
           child { node[b] {7} child{ node[c]{}} } child{ node[c]{}}}
           }
         ;
\end{tikzpicture}
}
\\
 (e) Second pairing pass, continued.
& (f) Final heap that results from an \op{Extract-Min}.
\end{tabular}
\end{center}
\caption{Illustration of how an \op{Extract-Min} is executed on a heap where the root has eight children.}
\label{f:em}
\end{figure}

A pairing heap is a heap-ordered general tree. Here a min-heap is
assumed. The basic operation on a pairing heap is the pairing
operation, which combines two pairing heaps into one by attaching the
root with the larger key value to the other root as its leftmost
child. Priority queue operations are implemented in a pairing heap as
follows: \op{Make-Heap} creates a new single node heap.  \op{Find-Min}
returns the data in the root of the heap.  \op{Meld} pairs the roots
of the two heaps.  \op{Insert} creates a new node and pairs it with the root of the
heap it is being inserted into.  \op{Decrease-Key} breaks off the node
and its induced subtree from the heap (if the node is not the root), decreases
the key value, and then pairs it with the root of the heap.
\op{Delete} breaks off the node to be deleted and its subtree,
performs an \op{Extract-Min} on the subtree, and pairs the
resultant tree to the root of the heap.  \op{Extract-Min} is the only
non-trivial operation. An \op{Extract-Min} removes and
returns the root, and then, in pairs, pairs the remaining trees in the
resultant forest.  Then, the remaining trees from right to left are
incrementally paired. See Figure~\ref{f:em} for an example of an \op{Extract-Min} executing on a pairing heap.  All pairing heap operations take constant
actual time, except \op{Extract-Min} and \op{Delete}, which take time
linear in the number of children of the node to be removed. For the
purposes of implementation, pairing heaps are typically stored as a binary tree
using the leftmost child, right sibling correspondence. Unless
otherwise stated, the standard tree terminology will refer to the
general tree representation.

\section{Main Result}

\subsection{Overview}

We claim that in a pairing heap the amortized runtimes of
\op{Find-Min}, \op{Make-Heap},  \op{Meld},  and \op{Insert} are $O(1)$
and \op{Decrease-Key}, \op{Delete} and \op{Extract-Min} are $O(\log
n)$. We adopt the convention
that $\log$ refers to the binary logarithm.

Let $X=x_1, x_2, \ldots , x_m$ be a sequence of operations to be
executed on an initially empty collection of heaps.  The remainder of the
notation used is implicitly parameterized by a fixed, but arbitrary,
$X$. Let $a_i$ be the actual time to execute operation $x_i$. Let
$n_i$ be the size of the heap that $x_i$ acted upon after the
execution of $x_i$. We define the actual cost of an operation to be
the number of pairings performed plus 1. This cost measure is chosen
as the number of parings clearly dominates the asymptotic runtime of
each non-constant-cost operation. Let $\hat{a}_i$ be the amortized
cost that we wish to prove for operation $x_i$. Specifically,  
if $x_i$ is an $\op{Insert}$ let $\hat{a}_i=\mhc$,
if $x_i$ is an $\op{Find-Min}$ let $\hat{a}_i=1$,
if $x_i$ is an $\op{Meld}$ let $\hat{a}_i=0$,
if $x_i$ is an $\op{Make-Heap}$ let $\hat{a}_i=\mhc$,
if $x_i$ is an $\op{Delete}$ let $\hat{a}_i=\dac+\dmc \log n_i$,
if $x_i$ is an $\op{Decrease-Key}$ let $\hat{a}_i=\dcac +\dcmc \log n_i$ and
if $x_i$ is an $\op{Extract-Min}$ let $\hat{a}_i=\emmc \log (n_i+1)+\emac$.

Given this notation, the main theorem can be stated. The
time to execute any sequence on an initially empty pairing heap is
bounded by the sum of the previously stated amortized times for each
operation.  Formally,

\begin{theorem}\label{bigth}
$
\sum_{i=1}^m a_i
\leq
\sum_{i=1}^m \hat{a}_i 
$.
\end{theorem}

The potential method is used to prove this theorem.  We define below a
potential function $\Phi = \langle \Phi_0, \Phi_1, \ldots , \Phi_m \rangle$ that is simply a sequence of
real numbers. The potential method may be summarized in the following lemma:

\begin{lemma} \label{potmethod} If there exists a sequence $\Phi$ such that for all $i$, $1 \leq
i \leq m$, $\hat{a}_i \geq a_i + \Phi_i -\Phi_{i-1}$ and $\Phi_m -
\Phi_0 \geq 0$ then $\sum_{i=1}^m \hat{a}_i \geq \sum_{i=1}^m a_i$.
\end{lemma}

This lemma, which summarizes the potential method, may be proved by
simple algebraic manipulation. More details on the potential method may
be found in \cite{amortized}.

The proof of Theorem~1 proceeds as follows. Section~\ref{sec:pot} is
devoted to defining the potential function~$\Phi$. This potential function is complex and has several components.  We then prove, in a
sequence of lemmas, that $\hat{a}_i \geq a_i + \Phi_i -\Phi_{i-1}$ for
each type of operation $x_i$. Finally, we state in a lemma that
$\Phi_m - \Phi_0 \geq 0$. These lemmas, according to Lemma~\ref{potmethod}
are sufficient to prove Theorem~\ref{bigth}. Analysis of \op{Insert} is not presented separately, as \op{Insert} is just a \op{Make-Heap} followed by a \op{Meld}. Similarly, 
the analysis of \op{Delete} is not presented separately, as \op{Delete}'s can be implemented as a \op{Decrease-Key} to negative infinity followed by a \op{Extract-min}.

\subsection{The potential function} \label{sec:pot}

For the analysis, a \emph{color}, black or white, is assigned to every node,
and a \emph{weight} is assigned to those nodes colored white. A node is \emph{black}
if it will remain in the forest of heaps at the end of execution of
sequence $X$, and \emph{white} otherwise. The color of a node never changes, since it is determined as a function of the entire fixed sequence of operations. We say that a white node has a weight of \emph{heavy} if the number
of white nodes in its left subtree in the binary representation is
at least the number of white nodes in its right
subtree. Roots are always heavy by this definition and every node in a
heap of size $n$ can have at most $\log n$ heavy children.  We
say that the weight of a white node that is not heavy is \emph{light}.

We say a node has been \emph{captured} if its parent is black.  A captured
node must have a \op{Decrease-Key} performed on it later in the
execution sequence before it is involved in any pairings. White
captured nodes must have \op{Decrease-Key} or a \op{Delete}
performed on them later in the execution sequence.

The \emph{node potential} of a white node is the sum of four components: rank
potential, weight potential, triple white potential, and capture
potential. Let $s(x)$ be the number of white nodes in the induced
subtree of $x$ in the binary representation. The \emph{rank potential} of a
white node $x$, $r(x)$, is $\rk \log s(x)$. If a node is white and has immediate
right and left siblings that are also white, then the node is referred
to as a \emph{triple white} and has a \emph{triple white potential} of $0$; if it is not a triple white it
has $\gpc$ units of triple white potential. White heavy nodes have a
\emph{weight potential} of $0$, and white light nodes have $\gpc$ units of
weight potential.  We assign captured nodes a \emph{capture potential} of $0$
and non-captured nodes $\gpc$ units of capture potential.  Black nodes
are defined to have 0 rank potential, weight potential and triple-white
potential.  Thus, the node potential of a black node is only comprised of
its capture potential. We also assign each heap a \emph{heap potential} which
is $\hpac-\hp \sum_{k=1}^n \log k $ where $n$ is the number of white nodes
in the heap (note that this could be negative in the middle of the execution of $X$, but is zero at the beginning and positive at the end. See Section~\ref{s:gl}).  The potential of a forest of heaps, $\Phi_i$, is the sum
of the node potentials of the nodes in the heaps and the heap potentials of the
heaps as a function of the state of the structure after the execution
of $x_i$.

It must be noted that the potential function exists for the sole
purpose of analysis. To implement a pairing heap neither the potential
function, nor its constituents such as node color, need to be stored.

The intuition behind this potential function comes from several places.
The rank potential is taken from the original pairing heap analysis \cite{pair}, and is identical to that used in splay trees \cite{splay}.
The notion of heavy and light was inspired by a similar idea used in skew heaps and their analysis \cite{skew}. 
The notion of a captured node is due to the fact that a pairing of a black and a white node where the black node has a smaller key value can be paid for by the resultant \op{Decrease-Key} or \op{Delete} operation that must be performed on the white node before it takes place in any additional pairings. The triple-white potential is a new invention, and was needed to handle a case where locally a mixture of black and white nodes transitions to all white nodes during \op{Extract-Min}; this case where the triple-white potential is vital appears as Case~4 of the local analysis of Lemma~\ref{l:em}.

The amortized cost of each operation is now calculated using this
potential function. For each operation it is proved that $\hat{a}_i
\geq a_i + \Phi_i -\Phi_{i-1}$, that is, the amortized cost of
an operation is at least the actual cost plus the change in
potential. In order to analyze the change in potential, we analyze the
change in each of the components that are summed to make
the potential function. For each operation, only a limited number of
clearly defined nodes may have their node potential change. These are the
nodes that have had their parents change, or white nodes that have had
a neighboring sibling or white descendant in the binary representation
change.

\subsection{Global loss} \label{s:gl}

\begin{lemma}
$\Phi_m-\Phi_0 \geq 0$
\end{lemma}

\begin{proof}
Since the structure is initially empty, $\Phi_0=0$. By definition,
after the execution of all $m$ operations in $X$, there are no white nodes
left in the collection of heaps.  Since the only possible negative
component of the potential function, the heap potential, requires white
nodes to be negative, $\Phi_m \geq 0$ and the lemma holds.
\qed
\end{proof}

\subsection{Make-heap}

\begin{lemma}
If $x_i$ is a \op{Make-Heap}, $a_i + \Phi_i -\Phi_{i-1} \leq \mhc $.
\end{lemma}

\begin{proof}

In a \op{Make-Heap}, the only change in potential is caused by the
introduction of the new one-node heap. The potential of all
existing nodes and heaps is unchanged. The actual cost and changes in potential can be accounted for as follows:

\begin{description}

\item[Actual cost:] 1. No pairings are performed.

\item[Change in rank potential:] 0. The newly created root, if white, has one
node in its induced subtree, and thus has $\rk \log 1=0$ units of rank
potential.  If the newly created root is black, black nodes by
definition have zero rank potential.

\item[Change in weight potential:] 0. If the newly created node is white, it
will be heavy and will have 0 units of weight potential.  If the newly
created node is black, it does not have weight potential.

\item[Change in triple white potential:] $\gpc$. The new node is not a triple-white.

\item[Change in capture potential:] $\gpc$. The new node is not captured.

\item[Change in heap potential:] $\hpac$. The new heap has a potential of $\hpac$.
\end{description}

Summing the actual cost and the change in potential yields the amortized
cost: $a_i + \Phi_i -\Phi_{i-1}=\mhc$.
\qed
\end{proof}

\subsection{Meld}

\begin{lemma}
If $x_i$ is a \op{Meld}, $a_i + \Phi_i -\Phi_{i-1} \leq 0 $.
\end{lemma}

\begin{proof}
The only nodes in a \op{Meld} that can change potential are the two old
roots, and the (former) leftmost child of the new root. We use $a$ and $b$ to
denote the number of white nodes in the heaps being melded. 

\begin{description}
\item[Actual cost:] 2, as 1 pairing is performed.

\item[Change in weight potential:] $\leq \gpc$. Only the old root that lost the pairing can have
its weight potential change.

\item[Change in triple white potential:] $\leq 0$. 
Only one node will have its left sibling change, which is a necessary condition for a change of triple-right. Let $x$ be the root that wins the pairing and let $y$ be the root that loses the pairing. Observe that $x$'s left child will now have $y$ as its left sibling. Observe that the change in triple-white potential $x$'s left child can only be negative,
since $x$'s left child is not a triple white before the beginning of this
operation.

\item[Change in capture potential:] $\leq 0$. No gain is possible
since no node escapes capture in a \op{Meld}. A loss is possible as the new leftmost child of the new root can be captured.
\end{description}

\noindent
The remainder of the analysis breaks into two cases.

\begin{description}
\item[Case 1:] At least one of the two heaps being melded only contains black nodes.

\begin{description}

\item[Change in rank potential:] 0. No white nodes have any change in their
white descendants.

\item[Change in heap potential:] -\hpac. The heap with all black nodes (or one of
them if both have only black nodes) has a heap potential of \hpac, while
the resultant heap has the same potential as the other heap.

\end{description}
\item[Case 2:] Each of the heaps being melded contain some white nodes.

\begin{description}

\item[Change in rank potential:] $\leq \rkk \log (a+b)$.
The only two nodes that could possibly change rank potential would be the two
roots, and their rank potential could rise to be at most $\rk \log (a+b)$ each.

\item[Change in heap potential:] $\leq -\hp \log (a+b) - \hpac$. 

The change in heap potential is given by the expression:

$ \hpac - \hp \sum_{i=1}^{a+b} \log i -\hpac + \hp \sum_{i=1}^{a} \log i -\hpac + \hp \sum_{i=1}^b \log i$

$= -\hpac - \hp \sum_{i=1}^{a+b} \log i + \hp \sum_{i=1}^{a} \log i + \hp \sum_{i=2}^b \log i$

$\leq - \hpac - \hp \sum_{i=1}^{a+b} \log i + \hp \sum_{i=1}^{a+b-1} \log i$

$\leq -\hpac -\hp \log (a+b).$
\end{description}
\end{description}

Summing the actual cost and the change in potential for both Case 1 and
Case 2 yields the same amortized cost: $a_i + \Phi_i -\Phi_{i-1}\leq  0$.
\qed
\end{proof}

\subsection{Decrease-key}

\begin{lemma}
If $x_i$ is a \op{Decrease-Key}, $a_i + \Phi_i -\Phi_{i-1} \leq \dcac +\dcmc \log n_i$.
\end{lemma}

Actual cost: 2, since one pairing is performed.

Rank Potential: $\leq \rk \log n_i$. The node on which the \op{Decrease-Key} is performed
could gain as most as $\rk \log n_i$ in rank potential. The former ancestors of the node in the binary representation could have their rank potential decrease.

Weight potential: $\leq \gpc  \log n_i  + 12$. The node the
decrease key is performed on can gain $\gpc$ in weight potential if it becomes heavy.  Also,
on the path $p$ in the binary representation from the node on which the \op{Decrease-Key} is to be
performed to the root, the removal of the node and its subtree may
cause some nodes to change their status from light to heavy or vice
versa.  Only changing from heavy to light causes a potential gain, and in such nodes the path $p$ goes through the node and its left child in the binary representation. After the \op{Decrease-Key} is performed there are only $\log n$ such nodes  because 
the $s(\cdot)$ values on any ancestor-descendent path in the binary representation is non-increasing in general,
and decreases by at least a factor of two from a light node to its left child in the binary representation. 
Thus this gain of $\gpc$ can only happen in at most $ \log n_i +1$ nodes. 

Capture potential: $\leq \gpc$.  There can only be a change of $\gpc$
in capture potential if the node which the \op{Decrease-Key} is
performed on was captured. 

Triple white: $\leq \gpc$. Among the node on which the \op{Decrease-Key} is
performed, and its two former siblings to the left and right, a total
of $\gpc$ units of triple white potential can be gained.  The former left child of the new root, will now be the second-leftmost child and could become a triple-white, thus releasing $\gpc$ units of potential.

Summing the actual cost and the change in potential yields the amortized
cost: $a_i + \Phi_i -\Phi_{i-1}\leq \dcac + \dcmc \log n_i$.

\subsection{\op{Extract-Min}} \label{sec:em}

\begin{lemma}\label{l:em}
Amortized cost of \op{Extract-Min} is $\emmc \log (n_i+1) + \emac$.
\end{lemma} 
\begin{proof}
Actual cost: If there are $c$ children of the root, the actual cost is $c$, since $c-1$
pairings are performed. 

As the analysis of \op{Extract-Min} is long, it is split into several
parts.  The first part examines the change of heap potential. Next, we
note that the only nodes that can have their node potential change are
the old root which is removed, the children of the old root, and the
grandchildren of the old root. The second part bounds the changes of
node potential of the old root. The third part bounds the changes of
triple-white potential of the grandchildren of the root, since this is
the only type of node potential change possible in these nodes.  The
analysis of the node potential change in the children of the old root
is the most complicated step, and it is presented in two parts: global
and local. In the global part we analyze some potential changes over
the whole structure, while in the local section we analyze how blocks
of six children of the root are affected by the remaining potential
changes. In the global section we analyze changes in rank potential,
gains in weight potential, and gains of capture potential. Part of this analysis, which is a variation of that of that which can be found in the original pairing heap paper \cite{pair} and mimimcs the original analysis of splay trees \cite{splay}, is presented separately in Section~\ref{rankchange}. In the
local section we analyze losses of weight potential, and changes of
triple white potential and losses of capture potential. We let $w$ denote the number of white-white parings in the first pairing pass (Nodes are colored. Pairings involve two nodes. A white-white pairing is a pairing of two white nodes).

\begin{description}

\item[Heap potential:] $\hp \log (n_i+1)$. This is caused by the heap's size
being reduced from $n_i+1$ to $n_i$.

\item[Node potential of root:] $\leq 0$. The removal of the root itself causes no
potential gain, since it has nonnegative potential

\item[Triple white potential of grandchildren of the old root:] $\leq 0$.  Some
grandchildren of the old root that are white, have white right
siblings and no left siblings may, through a pairing that their
parent is involved in, acquire a left white sibling and become a triple-white.
This can only cause a loss of potential.

\item[Global changes:]
\

\begin{description}

\item[Change of rank potential:] $\leq \rkkk \log n_i -  \rkk w$. 
The derivation of this is a variation of
the original pairing heap analysis of
\cite{pair}. This is topic of Section~\ref{rankchange} and appears as Lemma~\ref{l:ema}.

\item[Gains in weight potential:] $\leq \gpc \log n_i+\gpc$.  There are at most $\log
n_i+1$ heavy children of the root, since the $s(\cdot)$ value of the right sibling of a heavy node decreases by at least a factor of two, and so the potential gain caused by
heavy nodes becoming light is at most $\gpc \log n_i+\gpc$.

\item[Gains of capture potential:] 0.
The \op{Extract-Min} operation can cause
no increase in capture potential. An increase in the capture 
potential can only happen when a previously captured node becomes uncaptured.
However, since the root bas white, none of the children of the root
can have black parents, and thus none were captured.
\end{description}

\item[Local changes:]
In order to analyze other changes in potential (changes in triple
white potential, losses of weight potential caused by a node becoming
heavy, and changes in black nodes' potential) we break the $c$ children of
the root into blocks of six nodes, excluding the rightmost two
nodes. At most $\nbn$ nodes can not be included in this analysis, and
they could incur a potential gain of up to $\nbpe$ each for at total of $\nbp$---white node's triple white and weight are bounded by $\gpc$ each, and black nodes only have capture potential which is at most $\gpc$. In analyzing these
specific potential changes in each block of six, there are seven cases. See Figure~\ref{f:cases} for a decision tree showing how to determine which case applies.

Since for each case we are only considering losses in weight
potential and capture potential, in some parts we do not explicitly
state them and simply assume they are nonpositive. Triple white potential
must be considered in each part because we are considering
gains, as well as losses, however we show gains are limited to Case 1.

\newcommand{\ey}{edge from parent node[left] {yes} }
\newcommand{\eno}{edge from parent node[right] {no} }

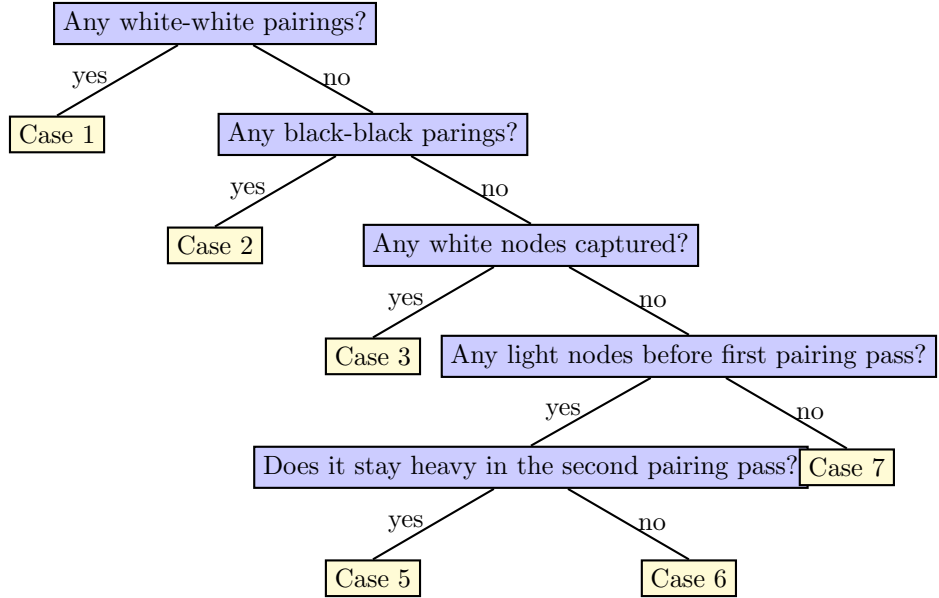
\begin{figure}
\begin{center}
\begin{tikzpicture}[child anchor=north,thick,scale=0.70,level distance=5pc]
   
    \tikzstyle{b}=[circle,draw,fill=blue!20]
        \tikzstyle{c}=[isosceles triangle,draw,fill=yellow,shape border rotate=90]]
        \tikzstyle{a}=[draw,fill=yellow!20]
        \tikzstyle{d}=[draw,fill=blue!20]
\tikzstyle{level 1}=[sibling distance=60mm]
    \node [d]{Any white-white pairings?}
        child { node[a]{Case 1} \ey }
        child { node[d]{Any black-black parings?}
          child{ node[a]{Case 2} \ey }
          child{ node[d]{Any white nodes captured?} 
            child{ node[a]{Case 3} \ey}
            child{ node[d]{Any light nodes before first pairing pass?}
              child{ node[d]{Does it stay heavy in the second pairing pass?} 
                child{ node[a]{Case 5} \ey}
                child{ node[a]{Case 6} \eno}
              \ey}
              child{ node[a]{Case 7} \eno}
            \eno }
          \eno} 
        \eno}
    ;
\end{tikzpicture}
\end{center}
\caption{Determination of which case applies in the local part of the proof of Lemma~\ref{l:em}.}
\label{f:cases}
\end{figure}

\begin{description}
\item[Case 1:] There is at least one white-white pairing in the first pairing
pass.

Triple white: $\leq \gpcnib$. The only gains in potential are the possible
gain of $\gpc$ units of triple white potential for each white involved in a
white-white pairing. This is the only case where a gain in potential,
among the components of the potential function under consideration is
possible. A gain in triple white potential can only occur when a child
of the former root is a triple white, and thus has white left and
right siblings. However, this guarantees that any triple white be
involved in a white-white pairing and this falls into this case. Thus,
increases in triple white potential can not occur in the following
cases and will not be considered.
\\ \\
Observe that if Case 1 does not apply all pairings in the first pairing pass are white-black or black-black.
\\ \\
\item[Case 2:] There is at least one
black-black pairing in the first pairing pass.

Capture potential: $\leq -\gpc$.
The black-black pairing(s) causes a loss of at least $\gpc$ units of
potential.
\\ \\
Observe that if Cases 1-2 do not apply then all pairings in the first pairing pass are black-white.
\\ \\
\item[Case 3:] At
least one of the three white nodes is captured.

Capture potential: $\leq -\gpc$. The capturing of the node(s) causes a
capture potential loss of at least $\gpc$.
\\ \\
Observe that if Cases 1-3 do not apply then all pairing in the first pairing pass are back-white, and the white nodes win all three first-pairing-pass pairings.
\\ \\
\item[Case 4:] All three white nodes participate in the second pairing pass and lose.

Triple white potential: $\leq -\gpc$.  Having all three loose the pairings
in the second pairing pass causes a loss of potential of $\gpc$, due to the
change of status of the middle white node to a triple white.
\\ \\
Observe that if Cases 1-4 do not apply then all pairing in the first pairing pass are back-white, and the white nodes win all three first-pairing-pass pairings and at least one of them win a second-pairing pass pairing. Case 5-7 are based on the weight of such a node. 
\\ \\
\item[Case 5:] Of one of the three white three nodes that participate in the second pairing pass
and win, its weight is light before the operation and heavy after both pairing passes.

Weight potential: $\leq -\gpc$
The light node becomes heavy, as all nodes previously on its right are
now in its subtree. Additional nodes that were to the node's left may
also be added to its subtree, but this just makes it more heavy.  This
causes a loss of $\gpc$ units of weight potential.

\item[Case 6:] Of one of the three white three nodes that participate in the second pairing pass
and win, its weight is light before the operation and heavy after the first pairing pass, and light after the second pairing pass.

We may assume no potential gain. This case can only happen $\log n_i$ times,
because in between the first and second pairing passes there are at most $\log n_i$ heavy roots.

\item[Case 7:] Observe that if Cases 1-6 do not apply then all pairing in the first pairing pass are back-white, and the white nodes win all three first-pairing-pass pairings and at least one of them win a second-pairing pass pairing, and that all such nodes are heavy at the beginning of the operation.

We may assume no potential gain. This case can only happen $\log n_i$ times,
because there are at most $\log n_i$ heavy children of the root.

\end{description}

\end{description}

We now sum together the potential changes covered by these seven cases.
Each application of Case~1 causes a potential increase of $\gpcnib$, and covers at least one white-white pairing. Thus the total gain caused by Case~1 is at most $\gpcnib w$. Since there are at least $\lfloor \frac{c-2}{\nib} \rfloor \geq \frac{c}{\nib}-2$ blocks of \nib, and none of Cases~2-6 have white-white pairings in the first pairing pass, there are at least $\frac{c}{\nib} -2 - w$ blocks of six covered by Cases~2-6.
Each of cases 2-5 has a loss of at least $6$, and Cases~6-7 have no gain. Since Cases~6 and~7 can each only happen in $\log n_i$ blocks, that means the total potential loss of Cases~2-6 is at least $\gpc(\frac c{\nib} -2 - w -2\log n_i)$.

We now summarize the potential changes discussed, and bring together all parts of our analysis to bound the total potential gain as at most:

\begin{align*}
\underbrace{\hp\log (n_i+1)}_{\text{Heap potential}}+
\underbrace{
\overbrace{\rkkk \log n_i - \rkk w}^{\text{Rank}}+ 
\overbrace{\gpc \log n_i +\gpc}^{\text{Weight}}
}_{\text{Global}}
\\+
\underbrace{
\overbrace{\nbp}^{{\text{Not in block}}}+
\overbrace{\gpcnib w}^{\text{Case~1}}-
\overbrace{\gpc\left(\frac c{\nib} -2 - w -2\log n_i\right)}^{\text{Cases~2-7}}
}_{\text{Local}}.
\end{align*}

This can be simplified, observing $-\rkk w + \gpcnib w + \gpc w=0$, to give an upper bound on the potential gain of
$\emmc \log (n_i+1) -c +\emacb$.
Then summing the actual
cost of the \op{Extract-Min} operation, $c-1$ with the maximum potential gain
yields the amortized
cost: $a_i + \Phi_i -\Phi_{i-1} \leq \emmc \log (n_i+1)+\emac$.
\qed
\end{proof}

\subsection{Change in rank potential}
\label{rankchange}

\newcommand{\ssbb}[1]{\begin{center}\scalebox{0.9}{#1}\end{center}}

\newcommand{\btb}{
\begin{tikzpicture}[child anchor=north,thick,scale=0.70]
    \tikzstyle{a}=[circle]
    \tikzstyle{b}=[circle,draw,fill=blue!20,inner sep=1pt]
    \tikzstyle{d}=[circle,draw,fill=blue!20,inner sep=3pt]
        \tikzstyle{c}=[isosceles triangle,draw,fill=yellow,shape border rotate=90]]
            \tikzstyle{w}=[fill=white]
    \tikzstyle{k}=[fill=black,text=white]

}

\newcommand{\pairb}[2]{
 \draw[red,<->,draw, thick,shorten >=2pt, shorten <=2pt] (#1) .. controls +(up:1cm) and +(up:1cm) .. node[above,sloped] {pair} (#2);
}

\begin{figure}
\begin{center}
\begin{tabular}{p{2in}p{2in}}
\ssbb{
\btb
\tikzstyle{level 1}=[sibling distance=13mm]
    \node[a] {}
    	   child { node[c]{}edge from parent[draw=none] }
        child { node[a]{$\cdots$} edge from parent[draw=none]}
        child { node[b] {$a_{\eta}$} child{ node[c]{}} edge from parent[draw=none]}
        child { node[b] {$b_{\eta}$} child{ node[c]{}} edge from parent[draw=none]}
        child { node[b] {$c_{\eta}$} child{ node[c]{}} edge from parent[draw=none]}
        child { node[a]{$\cdots$} edge from parent[draw=none]}
    	   child { node[c]{}edge from parent[draw=none]}
    ;
\end{tikzpicture}
}
&
\ssbb{
\btb
\tikzstyle{level 1}=[sibling distance=13mm]
\tikzstyle{level 2}=[sibling distance=11mm]
    \node[a] {}
    	   child { node[c]{}edge from parent[draw=none]}
        child { node[a]{$\cdots$} edge from parent[draw=none]}
        child { node[b] {$w_{\eta}$} 
          child { node[b] {$l_{\eta}$} child{ node[c]{}} }
          child{ node[c]{}} edge from parent[draw=none]}
        child { node[b] {$c_{\eta}$} child{ node[c]{}} edge from parent[draw=none]}
        child { node[a]{$\cdots$} edge from parent[draw=none]}
    	   child { node[c]{}edge from parent[draw=none]}
    ;
\end{tikzpicture}
}
\\
(a) Before pairing, general heap view.&
(b) After pairing, general heap view.
\\
\ssbb{
\btb
\node[a]{}
  child {node [a]{} edge from parent[draw=none]}
  child {node [b]{$a_{\eta}$ } 
   child { node[c]{} }
   child { node[b]{$b_{\eta}$} 
     child { node[c]{} }
     child { node[b]{$c_{\eta}$} 
       child { node[c]{} }
      child { node[c]{} }    
     }}}
;
\end{tikzpicture}
}
&
\ssbb{
\btb
\tikzstyle{level 2}=[sibling distance=20mm]
\tikzstyle{level 3}=[sibling distance=10mm]
\node[a]{}
  child {node [a]{} edge from parent[draw=none]}
  child {node [b]{$w_{\eta}$}
       child {node [b]{$l_{\eta}$}     
         child { node[c]{} }
         child { node[c]{} }
       }
       child {node [b]{$c_{\eta}$}     
         child { node[c]{} }
         child { node[c]{} }
       }       
  }
;
\end{tikzpicture}
}
\\
(c) Before pairing, binary view.&
(d) After pairing, binary view.
\end{tabular}
\caption{Illustration of the notation used to analyze a single pairing. Nodes
$a_{\eta}$ and $b_{\eta}$ are paired, and $w_{\eta}$ wins the pairing and $l_{\eta}$ loses the pairing.
Pre- and post-pairing pictures are presented, in both the general tree and binary views.
}
\label{f:onepair}
\end{center}
\end{figure}

In this subsection we examine the change in rank potential in an
\op{Extract-Min} operation, which was postponed in Section~\ref{sec:em}.
The proof of this bound follows in spirit the much simpler related proof in the original pairing heap paper~\cite{pair}, but is adjusted to take into account the presence of node coloring.

Pairings in an \op{Extract-Min} are always performed between a node and its right sibling. Let $\eta$ denote the left node in the pairing, and we adopt the terminology that the pairing is performed on $\eta$. For clarity we introduce notion to separately represent the two nodes before and after the pairing. 
This notation is described now and illustrated in Figure~\ref{f:onepair}.
Let  $a_\eta = \eta$ and $b_\eta$ be the right sibling of $a_\eta$; the nodes $a_\eta$ and $b_\eta$ are paired in the first pairing pass. 
We will refer to these nodes after the pairing has been completed as $w_\eta$ and $l_\eta$; the winner of the pairing is $w_\eta$ and the loser is $l_\eta$. The rank potential increase of the pairing is then simply $r(w_\eta)+r(l_\eta)-r(a_\eta)-r(b_\eta)$. We use $c_\eta$ to denote the right sibling of $b_\eta$ before the pairing; if there is no such sibling, we imagine $c_\eta$ to be a lone black node. 

Given this notation, we now bound potential change caused by a single pairing on a node $\eta$.

\begin{lemma} \label{lem:wpair}
In a single pairing 
the rank potential change is at most $\rkk \log s(a_\eta)- \rkk \log s(c_\eta) -\rkk$ if both nodes in the paring are white.
\end{lemma}

\begin{proof}
This bound on the rank potential change can be mathematically derived as follows, starting from the ranks of the two items involved in the pairing:

\begin{align}
&r(w_\eta)+r(l_\eta)-r(a_\eta)-r(b_\eta) \nonumber
\\
\intertext{Observe that $r(a_\eta)=r(w_\eta)$, and cancel them.}
&= r(l_\eta)-r(b_\eta) \label{eqa}
\\
\intertext{Use the definition that $r(x)=\rk \log s(x)$.}
&= \rk \log s(l_\eta) - \rk \log s(b_\eta)
\nonumber
\\
\intertext{Add and subtract $\rk \log s(c_\eta) + \rkk$.}
&= \rk \log s(l_\eta) + \rk \log s(c_\eta) + \rkk - \rk \log s(b_\eta) - \rk \log s(c_\eta) - \rkk
\nonumber
\\
\intertext{Use the fact that $s(b_\eta)\geq s(c_\eta)$.}
&= \rk \log s(l_\eta) + \rk \log s(c_\eta) + \rkk - \rkk \log s(c_\eta) - \rkk
\nonumber
\\
\intertext{Combine first three terms.}
&\leq \rk \log 4 s(l_\eta) s(c_\eta)- \rkk \log s(c_\eta) - \rkk
\nonumber
\\
\intertext{Use the fact that $(x+y)^2 \geq 4xy$.}
&\leq \rk \log (s(l_\eta)+s(c_\eta))^2- \rkk \log s(c_\eta) - \rkk
\nonumber
\\
&\leq \rkk \log (s(l_\eta)+s(c_\eta))- \rkk \log s(c_\eta) - \rkk
\nonumber
\intertext{Observe that $s(l_\eta)+s(c_\eta)\leq s(a_\eta)$}
\nonumber
\nonumber
&\leq \rkk \log s(a_\eta)- \rkk \log s(c_\eta) - \rkk
\nonumber
\end{align}

\
\qed
\end{proof}

\newcommand{\pairfig}[4]{
\ssbb{
\btb
\node[a]{}
  child {node [a]{} edge from parent[draw=none]}
  child {node [b,#1]{$a_{\eta}$ } 
   child { node[c]{} }
   child { node[b,#2]{$b_{\eta}$} 
     child { node[c]{} }
     child { node[b]{$c_{\eta}$} 
       child { node[c]{} }
      child { node[c]{} }    
     }}}
;
\end{tikzpicture}
}
&
\ssbb{
\btb
\tikzstyle{level 2}=[sibling distance=20mm]
\tikzstyle{level 3}=[sibling distance=10mm]
\node[a]{}
  child {node [a]{} edge from parent[draw=none]}
  child {node [b,#3]{$w_{\eta}$}
       child {node [b,#4]{$l_{\eta}$}     
         child { node[c]{} }
         child { node[c]{} }
       }
       child {node [b]{$c_{\eta}$}     
         child { node[c]{} }
         child { node[c]{} }
       }       
  }
;
\end{tikzpicture}
}
}

\begin{figure}
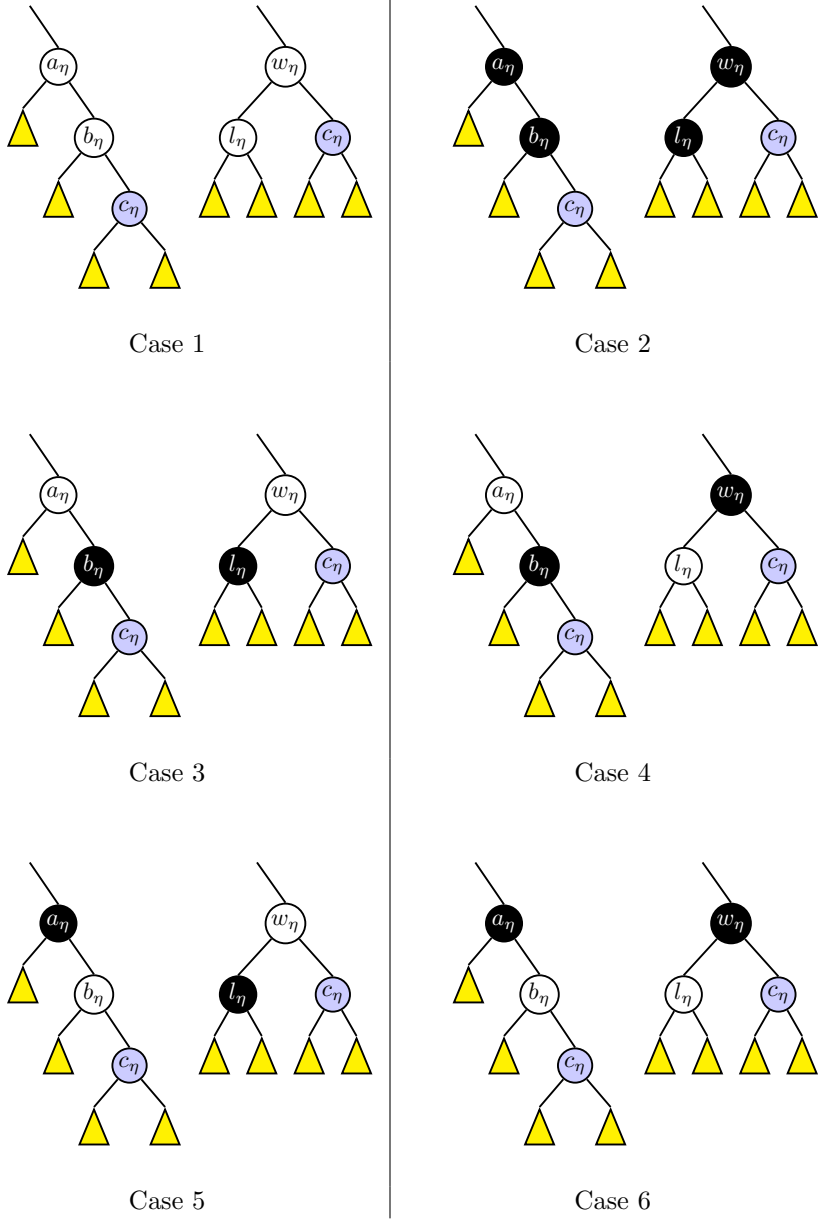

\begin{center}
\begin{tabular}{p{1in}p{1in}|p{1in}p{1in}}
\pairfig{w}{w}{w}{w} &
\pairfig{k}{k}{k}{k} \\
\multicolumn{2}{c|}{Case 1} &
\multicolumn{2}{c}{Case 2} \\
\pairfig{w}{k}{w}{k} &
\pairfig{w}{k}{k}{w} \\
\multicolumn{2}{c|}{Case 3} &
\multicolumn{2}{c}{Case 4} \\
\pairfig{k}{w}{w}{k} &
\pairfig{k}{w}{k}{w} \\
\multicolumn{2}{c|}{Case 5} &
\multicolumn{2}{c}{Case 6} \\
\end{tabular}
\end{center}
\caption{The six cases in the proof of Lemma~\ref{lem:onepair}, in the binary view. The left figure of each case is before pairing $a_{\eta}$ and $b_{\eta}$, and the right is after with $w_{\eta}$ denoting the winner of the pairing and $l_{\eta}$ denoting the loser.}
\label{f:onepaira}
\end{figure}

\begin{lemma} 
\label{lem:onepair}
In a single pairing 
the rank potential change is at most $\rkk \log s(a_\eta)- \rkk \log s(b_\eta)$. 
\end{lemma}

\begin{proof}
The analysis is split into six different types of pairing: white-white, black-black, and four types of black-white depending on whether the white node starts as $a_\eta$ or $b_\eta$ and ends as $w_\eta$ or $l_\eta$; see Figure~\ref{f:onepaira}. We look at each of these cases separately.

\paragraph{Case 1: White-white.} From \eqref{eqa} of Lemma~\ref{lem:wpair} we know the rank potential gain is at most $r(l_\eta)-r(b_\eta)$; observing that $r(l_\eta) \leq r(a_\eta)$ bounds the rank potential gain in this case to be at most $\rkk \log s(a_\eta)- \rkk \log s(b_\eta)$.

\paragraph{Case 2: Black-black.} None of the nodes involved have rank potential and thus there is no rank potential change. As $0\leq  \rkk \log s(a_\eta)- \rkk \log s(b_\eta)$, that gives the result in this case.

\paragraph{Case 3: White-black, white is $a_\eta$ and $w_\eta$.} Since $r(a_\eta)=r(w_\eta)$, there is no potential change. As $0\leq  \rkk \log s(a_\eta)- \rkk \log s(b_\eta)$, that gives the result in this case.

\paragraph{Case 4: White-black, white is $a_\eta$ and $l_\eta$.} Since $r(a_\eta)\geq r(l_\eta)$, there is no potential gain. As $0\leq  \rkk \log s(a_\eta)- \rkk \log s(b_\eta)$, that gives the result in this case.

\paragraph{Case 5: White-black, white is $b_\eta$ and $w_\eta$.} The gain is $r(w_\eta)-r(b_\eta)$. 
Since $r(w_\eta)=\rk \log s(w_\eta)=\rk \log s(a_\eta)$, the rank potential gain is at most
$\rkk \log s(a_\eta)- \rkk \log s(b_\eta)$.

\paragraph{Case 6: White-black, white is $b_\eta$ and $l_\eta$.} The gain is is $r(l_\eta)-r(b_\eta)$. 
As in the white-white case, observing that $r(l_\eta) \leq r(a_\eta)$ bounds the rank potential gain in this case to be at most $\rkk \log s(a_\eta)- \rkk \log s(b_\eta)$.
\qed
\end{proof}

\begin{corollary}
\label{co}
In a single pairing 
the rank potential change is at most $\rkk \log s(a_\eta)- \rkk \log s(c_\eta)$. 
\end{corollary}

\begin{proof}
Follows from Lemma~\ref{lem:onepair} and the fact that $s(c_\eta)\leq s(b_\eta)$.
\qed
\end{proof}

\begin{lemma} \label{l:ema}
The rank potential gain in the \op{Extract-Min} operation where there are $w$ white-white pairings in the first pairing pass is at most $\rkkk \log n_i -\rkk w$.
\end{lemma} 

\begin{proof}
In an  \op{Extract-Min}, the root is removed and returned, and then there are two pairing passes. We look at the potential change of each of these events separately. Removing the root causes no gain in potential. 

\paragraph{First pairing pass.}
In the first pairing pass, let $p$ be the total number of pairings and let $\eta_j$ denote the item the $j$th pairing from the left is performed on.  Lemma~\ref{lem:wpair} and Corollary~\ref{co}  bound the potential gain of the first pairing pass to be:

\begin{align*}
 &
 \overbrace{\mathop{\sum_{j\in [1,p]}}_{\mathclap{\eta_i\text{ is white-white}}} (\rkk \log s(a_{\eta_j})- \rkk \log s(c_{\eta_j}) -\rkk)}^{\text{from Lemma~\ref{lem:wpair}}}
+ \overbrace{\mathop{\sum_{j\in [1,p]}}_{\mathclap{\eta_i\text{ is not white-white}}} (\rkk \log s(a_{\eta_j})- \rkk \log s(c_{\eta_j}))}^{\text{from Corollary~\ref{co}}}
\\
\intertext{Rearranging and using the fact that the left sum has exactly $w$ elements.}
& =- \rkk w + \sum_{j=1}^p (\rkk \log (s(a_{\eta_j}))- \rkk \log s(c_{\eta_j})) 
\\
\intertext{Since $c_{\eta_j}$ is $a_{\eta_{j+1}}$ the sum telescopes.}
& =- \rkk w + \rkk \log s(a_{\eta_1})- \rkk \log s(c_{\eta_p})) 
\\
& \leq - \rkk w + \rkk \log s(a_{\eta_1})\\
& \leq - \rkk w + \rkk \log n_i
\end{align*}

\paragraph{Second pairing pass.}
In the second pairing pass, let $q$ be the total number of pairings, and let $\mu_j$ denote the item that the $j$th pairing is performed on. As the pairings are performed incrementally from right-to-left, $\mu_j$ is the $j+1$st node from the right before the beginning of the second pairing pass. Using Lemma~\ref{lem:onepair} bounds the potential gain of the second pairing pass to be:

$$\sum_{j=1}^q (\rkk \log s(a_{\eta_j})- \rkk \log s(b_{\eta_j})) $$

Since $b_{\eta_j}$ is $\eta_{j-1}$ the sum telescopes to 
at most $\rkk \log s(a_{\eta_q})$ which is at most $\rkk \log n_i$.

\paragraph{Putting it all together.}
Summing the upper bounds on rank potential gain of $  \rkk \log n_i- \rkk w$ for the first pairing pass and $\rkk \log n_i$ for the second pairing pass gives the claimed bound of a rank potential gain of at most 
$  \rkkk \log n_i- \rkk w$ for the \op{Extract-Min} operation.
\qed
\end{proof}

\section{Acknowledgements}

The author would like to thank Michael L.~Fredman for his support and advice, without which this work would not have been possible. He introduced the author to pairing heaps and the fun world of the amortized analysis of data structures using potential functions. Crucially, he pointed out that constant-amortized-time insertion in pairing heaps was not obvious and was worth figuring out.
I also thank Adam Buchsbaum, Va\v{s}ek Chv\'{a}tal, and Diane L.~Souvaine for their comments and encouragement on the preliminary version of this result which appeared in \cite{thesis}.



\end{document}